\documentclass[runningheads]{llncs}
\usepackage[T1]{fontenc}
\usepackage{amsfonts}
\usepackage{graphicx}
\usepackage{amsmath}
\usepackage{mathtools}
\usepackage{orcidlink}
\usepackage{amssymb}
\usepackage{algorithm}
\usepackage[noend]{algpseudocode}
\usepackage{booktabs}
\usepackage{pifont}
\usepackage{array}
\usepackage{verbatim}

\definecolor{darkred}{rgb}{0.5,0.0,0.0}

\newcommand{\aS}{\mathcal{X}}
\newcommand{\cS}{x}

\newcommand{\aY}{\mathcal{Y}}

\newcommand{\cW}{w}

\newcommand{\cond}{\varphi}
\newcommand{\conti}{\varepsilon}

\newcommand{\power}{\mathcal{P}}

\newcommand{\cert}{B}
\newcommand{\env}{f}
\newcommand{\sys}{\mathcal{F}}

\newcommand{\hori}{h}

\newcommand{\orac}{\mathbf{A}}
\newcommand{\moni}{\mathbf{M}}
\newcommand{\veri}{\mathbf{V}}

\newcommand{\RN}{\mathbb{R}}
\newcommand{\pRN}{\mathbb{R_+}}
\newcommand{\NN}{\mathbb{N}}
\newcommand{\pNN}{\mathbb{N}_+}

\begin{document}
\title{Formal Verification of Neural Certificates\\ Done Dynamically}
%
%
\author{
Thomas A. Henzinger\orcidlink{0000-0002-2985-7724}\and
Konstantin Kueffner\orcidlink{0000-0001-8974-2542}\and
Emily Yu\orcidlink{0000-0002-4993-773X} 
}
\authorrunning{T. Henzinger et al.}
%
\institute{Institute of Science and Technology Austria\\
\email{\{tah, kkueffner, zyu\}@ist.ac.at}}
\maketitle              
\begin{abstract}
Neural certificates have emerged as a powerful tool in cyber-physical systems control, providing witnesses of correctness. These certificates, such as barrier functions, often learned alongside control policies, once verified, serve as mathematical proofs of system safety. However, traditional formal verification of their defining conditions typically faces scalability challenges due to exhaustive state-space exploration. To address this challenge, we propose a lightweight runtime monitoring framework that integrates real-time verification and does not require access to the underlying control policy. Our monitor observes the system during deployment and performs on-the-fly verification of the certificate over a lookahead region to ensure safety within a finite prediction horizon.
We instantiate this framework for ReLU-based control barrier functions and demonstrate its practical effectiveness in a case study. Our approach enables timely detection of safety violations and incorrect certificates with minimal overhead, providing an effective but lightweight alternative to the static verification of the certificates.

\end{abstract}
\section{Introduction}
Forthcoming intelligent systems are increasingly integrated into industries and numerous application domains, including autonomous driving and medical image processing~\cite{abiodun2018state,anwar2018medical}. 
For example, deep reinforcement learning enables the automated synthesis of neural network controllers to address complex control tasks~\cite{kiumarsi2017optimal}.
However, in these safety-critical domains, ensuring the safety and correctness of machine-learned components presents significant challenges. 
Neural networks often lack transparency and explainability, undermining their trustworthiness. Without formal safety guarantees, the reliability and operation of cyber-physical systems remain constrained, thereby affecting their deployment in practice.

Certificate functions serve as mathematical proofs to establish the correctness of controllers. A certificate function~\cite{DBLP:journals/trob/DawsonGF23} is a mathematical mapping from system states to real values, where the satisfaction of its defining conditions guarantees that a desired system property holds. Notable examples include Lyapunov functions~\cite{DBLP:journals/automatica/Smith95}, used to prove stability with respect to a fixed point, and barrier functions~\cite{DBLP:journals/tac/PrajnaJP07}, which are employed to certify safety by characterizing forward-invariant sets. While the theory of Lyapunov functions has been extensively studied over the past several decades, recent advances leverage reinforcement learning to synthesize such certificates in the form of neural networks. A growing body of work in learning-based control explores the joint synthesis of both the certificate function and the control policy. We refer to~\cite{DBLP:journals/trob/DawsonGF23} for a more comprehensive overview of learning-based control using certificates.

Since certificate functions are often represented as neural networks, formal verification is typically employed to ensure their correctness. This gave rise to the learner-verifier framework~\cite{chang2019neural,DBLP:conf/tacas/ChatterjeeHLZ23,peruffo2021automated}, which is a synthesis framework to obtain valid neural certificates. This framework follows the style of Counterexample-Guided Inductive Synthesis (CEGIS). It consists of two main components: \emph{the learner}, which synthesizes both the control policy and the certificate function, and \emph{the verifier}, which checks whether the certificate satisfies its formal defining conditions. If the verifier identifies a counterexample, it is incorporated into the training set to refine the neural networks. This iterative loop continues until the certificate is successfully verified or a predefined termination criterion, such as a timeout, is met.

While formal verification offers correctness guarantees, it often suffers from scalability limitations due to its computational complexity, restricting its use in more complex control tasks. To overcome this challenge, \emph{runtime monitoring} has emerged as a promising complementary approach. A runtime monitor is a lightweight software module that operates in parallel to the controlled system, issuing warnings upon detecting violations of specified properties or certificate defining conditions. 
The Simplex architecture~\cite{crenshaw2007simplex,seto1998simplex}, for instance, employs a verified backup controller to override unsafe actions. 
Recent work has proposed a learner-monitor framework~\cite{yu2025neural}, analogous to the learner-verifier paradigm, in which the monitor continuously collects counterexamples during execution and incorporates them into the training data to refine the learned models.

\begin{figure}
    \centering
    \includegraphics[width=1\linewidth]{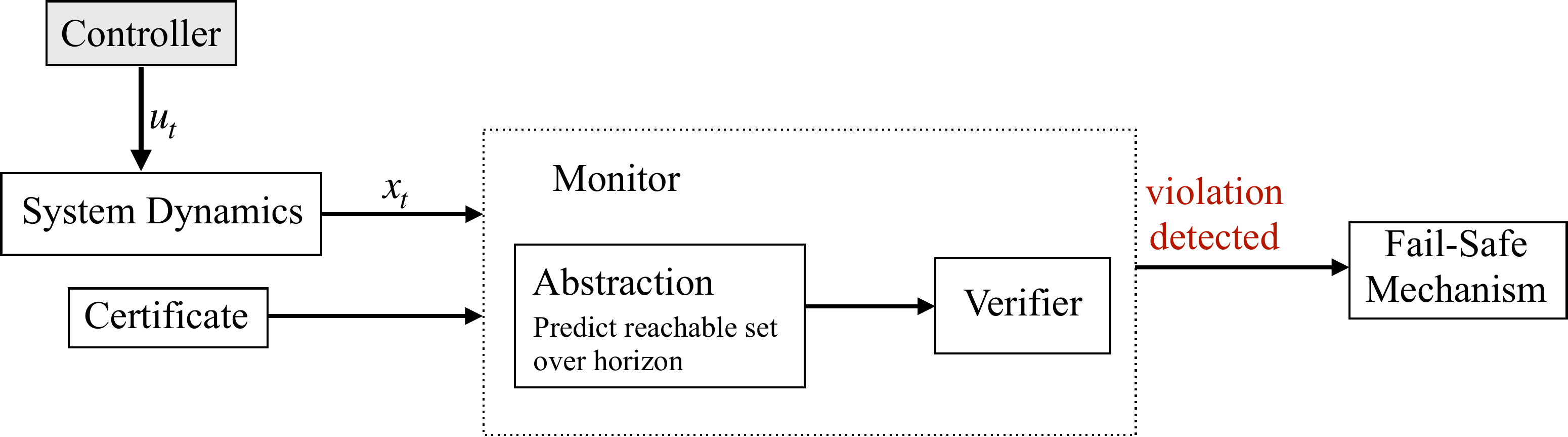}
\caption{Overview of our monitoring framework. At runtime, the monitor observes system states and computes an overapproximation of the reachable region using a local abstraction function. We do not assume access to the controller. A verifier checks whether the certificate remains valid in the lookahead region. If a violation is detected, a warning is issued and a fail-safe mechanism can be triggered.}
    \label{fig:enter-label}
\end{figure}

In this paper, we propose a general runtime monitoring framework that integrates partial static verification over a lookahead horizon to ensure the safety of the controlled systems, as illustrated in Fig.~\ref{fig:enter-label}. The framework is agnostic to the underlying controller and does not require access to future control inputs. At each time step, a monitor observes the system trajectory and uses a local abstraction function to over-approximate reachable states within the lookahead horizon. A certificate verifier then checks whether the safety certificate remains valid over this predicted region. If a violation is detected, the monitor issues a warning and can trigger a fail-safe response. Because our monitor detects violations ahead of time, the fail-safe mechanism is activated \emph{before} leaving the certified region. 

We instantiate this general framework using a specific verifier for ReLU-based control barrier functions (CBFs), leveraging the piecewise linear structure of ReLU networks for efficient localized verification. Our case study on a satellite rendezvous task demonstrate that the monitor can detect certificate violations as well as ensuring runtime safety with minimal overhead.

Overall, our contributions are as follows.
\begin{enumerate}
    \item We introduce a general framework for online verification of certificates, which leverages a local abstraction function and a verifier as a sub-routine to enable local verification over a finite lookahead horizon.
    \item We instantiate the framework for ReLU-based control barrier functions, and present a novel online verification algorithm that operates over neuron activation patterns. The algorithm adaptively constructs and verifies regions of the state space encountered at runtime, enabling fast detection of violations of safety as well as certificate defining conditions. 
\end{enumerate}
To the best of our knowledge, this is the first runtime monitoring approach that performs formal, on-the-fly certificate verification for neural-based control.

\section{Preliminaries}
We denote the set of natural number as $\NN$ and the set of real numbers as $\RN$. Their positive counterparts are denoted as $\pNN$ and $\pRN$ respectively.

Given a state space $\mathcal{X} \subseteq \mathbb{R}^m$ and an action space $\mathcal{U} \subseteq \mathbb{R}^n$, we consider continuous-time, control-affine dynamical systems of the form
\[
\dot{x} = f(x) + g(x) u,
\]
where $x\in \mathcal{X}, u\in \mathcal{U}, f : \mathbb{R}^m \to \mathbb{R}^m$ and $g : \mathbb{R}^m \to \mathbb{R}^{m \times n}$. A control policy $\pi : \mathcal{X} \to \mathcal{U}$ defines a closed-loop system $\mathcal{F} : \mathcal{X} \to \mathcal{X}$ with dynamics
\begin{align}
    \label{eq:system}
    \quad \mathcal{F}(x) = \env(x)+g(x) \pi(x).
\end{align}

The safe control problem requires that, under a given control policy, the system must never enter an unsafe state along any execution trace originating from an initial state. 
To prove that the process $\mathcal{F}$ satisfies properties such as safety or reachability, we rely on certificates. A certificate is a function $\cert\colon \aS\to \RN$ assigning a numerical value to each state. To be a valid certificate, this function $\cert$ must satisfy some property dependent conditions over the state space. To prove safety, a barrier function is needed.

\begin{definition}[Control Barrier Function]
\label{def:barrier}
Let $\mathcal{X}_0 \subseteq \mathcal{X}$ be the set of initial states, and let $\mathcal{X}_u \subseteq \mathcal{X}$ be the unsafe states. A continuously differentiable function $B : \mathcal{X} \to \mathbb{R}$ is called a \emph{control barrier function} if the following conditions hold:
\begin{enumerate}
    \item $B(x) \geq 0$ for all $x \in \mathcal{X}_0$;
    \item $B(x) < 0$ for all $x \in \mathcal{X}_u$;
    \item $\nabla B(x)(f(x)+g(x)u)) + \alpha(B(x)) \geq 0$ for all $x \in \{ x \in \mathcal{X} \mid B(x) \geq 0 \}$.
\end{enumerate}
Here, $\alpha : \mathbb{R} \to \mathbb{R}$ is a class $\mathcal{K}$ function that is strictly increasing and $\alpha(0) = 0$~\cite{ames2014control}.
\end{definition}

The set $\mathcal{C}=\{x \in \mathcal{X} \mid B(x) \geq 0\}$ is referred to as the forward invariant set. The expression $\nabla B(x)(f(x) + g(x)u)$ denotes the Lie derivative of the barrier function $B$ along the system dynamics. If a valid barrier function exists for a given control policy $\pi$, then the system is guaranteed to be safe. A formal proof of this result is provided in~\cite{ames2014control}.

In this paper, we also consider {ReLU-based barrier functions}, where the barrier function is implemented as a feed-forward neural network with \( L \) layers. Each layer \( i \) has dimension size \( M_i \), and the activation function used is the Rectified Linear Unit (ReLU), defined as \( \sigma : \mathbb{R}^a \to \mathbb{R}^a \) with
\[
\sigma(x) = 
\begin{cases}
x & \text{if } x \geq 0, \\
0 & \text{otherwise}.
\end{cases}
\]
Each neuron in the network is indexed by its layer and position within the layer, denoted by the pair \( (i, j) \). Given an input \( x \), let \( z_i^j \) denote the pre-activation value at neuron \( (i, j) \). A neuron is said to be \emph{active} if \( z_i^j > 0 \), \emph{unstable} if \( z_i^j = 0 \), and \emph{inactive} otherwise. We denote by \( \mathcal{A} = (A_1, \ldots, A_L) \) an activation pattern, where each \( A_i \) represents the set of active neurons in layer \( i \). Similarly, let \( \mathcal{T} = (T_1, \ldots, T_L) \) represent the unstable neurons.

For a given activation pattern, evaluating the neural network becomes significantly simpler, as the weights and biases reduce to fixed linear terms. We define masked weights and biases with respect to \( \mathcal{A} \). For the first layer:
\[
\overline{W}_{1j}(\mathcal{A}) =
\begin{cases}
W_{1j} & \text{if } j \in A_1, \\
0 & \text{otherwise},
\end{cases}
\quad
\overline{r}_{1j}(\mathcal{A}) =
\begin{cases}
r_{1j} & \text{if } j \in A_1, \\
0 & \text{otherwise}.
\end{cases}
\]
This implies that the output of the \( j \)-th neuron in the first layer is simply \( \overline{W}_{1j}(\mathcal{A})^\top x + \overline{r}_{1j}(\mathcal{A}) \).
For deeper layers \( i > 1 \), we recursively define:
\[
\overline{W}_{ij}(\mathcal{A}) =
\begin{cases}
W_{ij}^\top \overline{W}_{i-1}(\mathcal{A}) & \text{if } j \in A_i, \\
0 & \text{otherwise},
\end{cases}
\quad
\overline{r}_{ij}(\mathcal{A}) =
\begin{cases}
W_{ij}^\top \overline{r}_{i-1}(\mathcal{A}) + r_{ij} & \text{if } j \in A_i, \\
0 & \text{otherwise}.
\end{cases}
\]

We write \( \overline{W}(\mathcal{A}) \) to denote the effective linear transformation of the final layer under the activation pattern \( \mathcal{A} \). This linearized structure captures the piecewise-linear behavior of the ReLU network within the region defined by \( \mathcal{A} \).

\begin{definition}[ReLU-based CBF]\label{def:relu}
Let \( \mathcal{X}(\mathcal{A}) \) denote the region of the input space induced by activation pattern \( \mathcal{A} \). A function \( B \) is a ReLU-based control barrier function if, for every activation pattern \( \mathcal{A} \), the region \( \mathcal{X}(\mathcal{A}) \cap \mathcal{C} \) does not intersect the unsafe set:
\[
\mathcal{X}(\mathcal{A}) \cap \mathcal{C} \cap \mathcal{X}_{\mathrm{unsafe}} = \emptyset.
\]
Moreover, there exists an activation pattern \( \mathcal{A} \) such that for all \( x \) on the boundary where \( B(x) = 0 \), the following hold:
\begin{enumerate}
    \item \( (\overline{W}_{i-1}(\mathcal{A}) W_{ij})^\top (f(x) + g(x)\pi(x)) \geq 0 \) for all \( (i, j) \in \mathcal{T}(x) \cap \mathcal{A} \);
    \item \( (\overline{W}_{i-1}(\mathcal{A}) W_{ij})^\top (f(x) + g(x)\pi(x)) \leq 0 \) for all \( (i, j) \in \mathcal{T}(x) \setminus \mathcal{A} \);
    \item \( \overline{W}(\mathcal{A})^\top (f(x) + g(x)\pi(x)) \geq 0 \).
\end{enumerate}
\end{definition}

This definition encodes the requirement that the barrier function $B$ does not decrease along the system dynamics at the certificate boundary where $B(x) = 0$~\cite{zhang2023exact}. This definition guarantees that $B$ serves as a valid control barrier function over the region defined by the activation pattern $S$, and that the system remains within the forward invariant set $\mathcal{C}$ if it starts there.

\section{Monitoring Framework}
Classically the validity of a certificate function, e.g., the barrier certificate in Definition~\ref{def:barrier}, is verified on the entire state space. This is expensive and wasteful, especially in high-dimensions. As the dimension of the state space increases, both, finding and validating a certificate becomes exponentially more expensive. At the same time the volume of all states visited by the process on an infinite run becomes negligible compared to the volume of the entire domain, implying that most of the work done during verification is unnecessary.

\paragraph{Certificates at Runtime.}
We use $\cond$ to denote an arbitrary certificate condition, e.g., the conditions in Definition~\ref{def:barrier}. The binary value of a certificate condition is determined by the closed-loop system $\sys$, the certificate $\cert$, and a subset of the state space $\aY\subseteq \aS$, i.e., $\cond(\sys, \cert, \aY) \in \{0,1\}$. In static verification we check the certificate condition over the entire state space $\aS$. This ensures that $\varphi$ is satisfied on the infinite trace $\cW\coloneqq (\cS_t)_{t\in \pRN}$ generated by $\sys$. Intuitively, we would expect that validating the certificate condition for a single trace, i.e., $\cond(\sys, \cert, \cW)$, should be significantly simpler than validating it for the entire state space $\cond(\sys, \cert, \aS)$. The reason why static verification focuses on the entire state space is because it may be impossible to obtain the trace $\cW$ at development time due to uncertainty. 
However, for short time horizons this uncertainty remains manageable. This is where static verification done dynamically shines. 

\begin{example}
    Consider $\aS\coloneqq [0,1]^d$ for $d\in \pNN$ and let $\cW\coloneqq (\cS_t)_{t\in \pRN}$. If we know that the certificate $\cert$ and the system $\sys$ are Lipschitz with a constant $L$, a naive approach to validating the certificate condition on $\aS$ up to a precision $c>0$ is to construct a $c/L$-grid \cite{DBLP:conf/tacas/ChatterjeeHLZ23}. The number of vertices in the grid increases exponentially in the dimension, i.e., for $\aS$ we require approximately $(L/c)^d$ vertices. By contrast, a $c/L$-grid over $\cW$ requires only $L/c$ vertices. Hence, verifying the trace is exponentially faster than verifying the entire domain. 
\end{example}

\paragraph{Certificate Monitor.}
Our objective is to construct a certificate monitor which detects certificate condition violations before they occur. 
Let $\cW\coloneqq (\cS_t)_{t\in \pRN}$ be the trace generated by $\sys$. 
The trace up to time $t\in \pRN$ is defined as $\cW_{[0,t)}\coloneqq (\cS_t)_{t\in [0,t)}$. The monitor can only observe a finite subset of this trace as given by its observation frequency parameter $\conti>0$, which we consider as given, e.g., the parameter may be a function of the monitor's hardware constraints. For every interval $[a,b)\subseteq \RN$ we denote $[a,b)_{\conti}\coloneqq \{a+k\cdot \conti \mid k\in \NN \land a+k\cdot \conti < b\}$ as the finite discretization of $[a,b)$ w.r.t. $\conti$. We denote the observable trace by the monitor as $\cW_{[0,t)}^{\conti}\coloneqq (\cS_t)_{t\in [0,t)_{\conti}}$.
For a given safety horizon $\hori\in \pRN$, e.g., the time required to stop a vehicle, we want a monitor $\moni\colon \aS^*\to \{0,1\}$ that maps the trace $\cW_{[0,t)}^{\conti}$ into a verdict in $\{0,1\}$. We call the monitor $\moni$ sound, iff a positive verdict guarantees that the certificate is satisfied on the trace $\cW_{[0,t+ \hori)}$, i.e., 
\begin{align}
    \label{eq:sound}
  \moni\left(\cW_{[0,t)}^{\conti}\right) = 1 \implies \cond\left(\sys, \cert, \cW_{[0,t+ \hori)}\right) = 1.
\end{align}
\begin{problem}
    For a given closed-loop system $\sys$ as defined in Equation~\ref{eq:system}, a certificate function $\cert$, a certificate condition $\cond$, and a safety horizon $\hori$ with epsilon $\epsilon$, construct a sound certificate monitor $\moni$.
\end{problem}

\subsection{Monitor Construction}
We show how to construct a sound certificate monitor using a local abstraction function and a certificate verifier as subroutines. 
The local abstraction function provides a sound overapproximation of how the process behaves within a given time horizon. 
This overapproximation is a bounded region, e.g., a cone, within the state space. The verifier checks whether the certificate condition is satisfied for every state within a given region of the state space. 
Our monitor simply invokes the verifier on the uncertainty region computed by the local abstraction function. 
We chose this modular construction because there exists a plethora of predictive monitoring tools which can be utilized as a local abstraction function~\cite{DBLP:journals/corr/abs-2412-16564}. Similarly, there exists a variety of certificate verification tools 
that can assess the validity of a certificate for a given region~\cite{DBLP:conf/tacas/ChatterjeeHLZ23,zhao2020synthesizing,zhao2022verifying}. Each tool has its own assumptions on the system, as a consequence our certificate monitor can adapt and improve depending on the available subroutines.

\paragraph{Abstraction and verification.}
We assume access to a local abstraction function $\orac\colon \aS^*\times \pRN \to \power(\aS)$ and a verifier $\veri\colon\power(\aS) \to \{0,1\}$ where $\power(\aS)\coloneqq \{\aY\mid \aY\subseteq \aS\}$ is the set of all subsets of $\aS$. 
First, the abstraction function maps the trace $\cW_{[0,t)}^{\conti}\in \aS^*$ up to time $t\in \pRN$ and a time horizon $\hori\in \pRN$ into 
a subset of the domain, i.e.,  $\orac(\cW_{[0,t)}^{\conti}, \hori) \subseteq \aS$. We are guaranteed that the process is contained within the output region for the specified time horizon, i.e.,
\begin{align}
    \label{eq:orac}
    \forall t\in [t, t+\hori) . \; \cS_t \in \orac(\cW_{[0,t)}^{\conti}, \hori).
\end{align}
Second, the verifier checks whether the certificate $\cert$ satisfies condition $\cond$ w.r.t.\ the process $\sys$ for a given region $\aY \subseteq\aS$, i.e., 
\begin{align}
    \label{eq:veri}
    \forall \aY \subseteq \aS .\; \veri(\aY) = 1 \iff \cond(\sys, \cert, \aY) = 1
\end{align}
In Section~\ref{sec:verification} we provide a concrete examples for both the abstraction function and the verifier. For now we treat both of them as black-boxes.

\paragraph{Local Verification.}
By combining the abstraction function and the verifier we construct a simple subroutine ensuring that the certificate condition is satisfied locally, i.e., the certificate remains valid within the immediate time horizon. 
Assume we are at time $t\in \pRN$, have observed the trace $\cW_{[0,t)}^{\conti}$ derived from the actual trace $\cW_{[0,t)}$, and decided for a time horizon $\hori_t\in \pRN$. 
We invoke the abstraction function $\orac$ with the observed trace $\cW_{[0,t)}^{\conti}$ and time horizon $\hori_t$, and verify the generated region using $\veri$
to guarantee the validity of the certificate up until time $t+\hori_t$, i.e., we compute 
\begin{align}
\veri( \orac(\cW_{[0,t)}^{\conti}, \hori_t)).
\end{align}
If the verifier outputs $1$ the validity for certificate is assured up until $t+\hori$, if the verifier outputs $0$ the certificate condition is violated for a state in the predicted region. Hence, it may violated on $\cW_{[t,t+h)} \subseteq \orac(\cW_{[0,t)}^{\conti}, \hori_t)$. This gives a local overapproximation of the certificate condition value. 
\begin{lemma}
    \label{lemma:local}
    If the abstraction function satisfies Condition~\ref{eq:orac} and the verifier satisfies Condition~\ref{eq:veri}, we are guaranteed that 
    \begin{align*}
        \veri( \orac(\cW_{[0,t)}^{\conti}, \hori_t))=1 \implies \cond(\sys, \cert, \cW_{[t,t+\hori_t)})=1 .
    \end{align*}
\end{lemma}
\begin{proof}
    Condition~\ref{eq:orac} ensures that the over-approximation computed by the abstraction function contains the trace up to time $t+\hori_t$, i.e., $\cW_{[t,t+\hori_t)}) \subseteq \mathcal{Z}\coloneqq \orac(\cW_{[0,t)}^{\conti}, \hori_t) $.
    Condition~\ref{eq:veri} ensures that $\veri(\aY) = 1  \implies \cond(\sys, \cert, \cW_{[t,t+\hori_t)}) = 1$.
\end{proof}

\paragraph{Soundness.}
The subroutine described above guarantees that the certificate condition remains satisfied within the specified time horizon. However, verifying on the fly requires time and cannot be done continuously. 
Therefore, constructing a sound monitor requires us to get the timing right. 
Assume we can start the verification procedure after having observed the current state. If we are able to compute the local guarantee, i.e., $  \veri( \orac(\cW_{[0,t)}^{\conti}, \hori_t))$, for a time horizon of $\hori_t\coloneqq 2\cdot \conti + \hori$ in less than $\conti$-time, then we can obtain a sound monitor defined in Algorithm~\ref{alg:monitor-abstract}.
This guarantees that the monitor raises a warning \emph{before} the system enters a region where the certificate is invalid, allowing the system to activate the fail-safe mechanism in time.

\begin{theorem}
    Given a state space $\aS$, a system $\sys$, and a certificate $\cert$.
    Let $\orac\colon \aS^*\times \pRN \to \power(\aS)$ be an abstraction function (satisfying Eq.~\ref{eq:orac}, $\veri\colon\power(\aS) \to \{0,1\}$ be a verifier (satisfying Eq.~\ref{eq:veri}), $\hori\in \pRN$ be a time horizon, $\conti\in \pRN$ be a observation frequency parameter. 
    If the the sub-routine \emph{\textsc{Next}} can be computed in less than $\conti$-time, the monitor defined by Algorithm~\ref{alg:monitor-abstract} is sound, i.e., 
    \begin{align*}
        \moni_{\hori}(\cW_{[0,t)}^{\conti}) = 1 \implies \cond(\sys, \cert, \cW_{[0,t+ \hori)}) = 1.
    \end{align*}
\end{theorem}
\begin{proof}
    The sequence of observation points is $(k\cdot \conti)_{k\in \NN}$ we prove this claim by induction over the observation points. 
    Let $k=0$. We assume we have already verified the first $\conti+ \hori$ time interval, if the verdict $v=0$ we have found a potential violation and we are done. Otherwise, we know from Lemma~\ref{lemma:local} that this implies that $\cond(\sys, \cert, \cW_{[0,\conti+ \hori)}))=1$.
    This implies that every point $s\in [0, \conti)$ is covered.
    The induction hypothesis is that at the beginning of every observation point  $k\in\pNN$, $v_{k-1}=1$ implies that $\cond(\sys, \cert, \cW_{[0,\conti\cdot k + \hori)}))=1$.
    Assume we are at observation point $k+1$.
    If $v_k=0$ we are done. If $v_k=1$, then by the induction hypothesis we know that $\cond(\sys, \cert, \cW_{[0,\conti\cdot (k+1) + \hori)}))=1$.
    We start computing \textsc{Next} for the time horizon $2\conti+\hori$, i.e. upon termination we know whether there exists a potential certificate violation on $[(k+1)\cdot \conti, (k+3)\cdot \conti + \hori)$ or not. 
    By assumption we know that executing \textsc{Next} requires less than $\conti$ time. Hence, before time step $k+2$ we have finished computing $v_{k+1}$, which if $v_{k+1=1}$ guarantees that $\cond(\sys, \cert, \cW_{[0,\conti\cdot (k+2) + \hori)}))=1$.
\end{proof}

\paragraph{Implementation.}
Algorithm~\ref{alg:monitor-abstract} describes the schematic dynamic verification routine. We initialize the monitor with \textsc{Init}, which creates the initial empty trace $\cW_0$ and sets the initial verdict $v_0$ to $1$. During runtime the monitor executes \textsc{Next} whenever a new input state $\cS\in \aS$ is observed. At execution $k\in \pNN$, it appends the observed state $\cS$ to the past trace creating $\cW_k$, which is used to perform the verification step over the uncertainty region provided by the abstraction function; checking whether there is a certificate violation in the future. 
The verdict of this step is stored in $v_{\mathrm{buffer}}$. 
If the previous verdict $v_{k-1}$ is $1$, i.e., no certificate violation was detected in some previous iteration, the verdict at iteration $v_k$ is set to $v_{\mathrm{buffer}}$.
Otherwise, the verdict remains $0$, because once a violation is detected, the monitor will consider the certificate condition violated.

\begin{algorithm}
\caption{Schematic Monitor $\moni_h$}
\label{alg:monitor-abstract}
\begin{algorithmic}[1]
\State \textbf{Given:} time horizon $\hori\in \pRN$, observation frequency parameter $\conti$, abstraction function $\orac\colon \aS^*\times \pRN \to \power(\aS)$, a verifier $\veri\colon\power(\aS) \to \{0,1\}$. 
\Function{Init}{\,}
    \State $\cW_0 \gets \cS; \; v_0 \gets 1$
\EndFunction
\Function{Next}{$\cS$}
    \State $\cW_k \gets \cW_{k-1} \cdot \cS $
    \State $v_{\mathrm{buffer}} \gets \veri( \orac(\cW_k, 2\cdot \conti+\hori))$
    \State $v_k \gets  v_{\mathrm{buffer}}$ \textbf{if} $v_{k-1} = 1$ \textbf{else} $v_k \gets 0$
    \State\Return $v_k$
\EndFunction
\end{algorithmic}
\end{algorithm}

\begin{remark}
    If the execution of \textsc{Next} requires more time, parallelization can be exploited to avoid this problem, i.e., we simply execute \textsc{Next} while accounting for the additional computation time on a new thread at the end of every control interval. If the computation time remains constant, we can obtain a similar guarantee. Moreover, the number of threads can be bounded based on that time.
\end{remark}

\medskip
\noindent

\section{Online Verification of ReLU-based Control Barrier Functions (CBFs)}\label{sec:verification}

Our monitoring approach is designed to ensure runtime safety and verify the correctness of a neural certificate, without requiring full state space verification or access to future control inputs. In this section, we instantiate the general monitoring framework for ReLU-based control barrier functions.

\begin{figure}
    \centering
    \includegraphics[width=\linewidth]{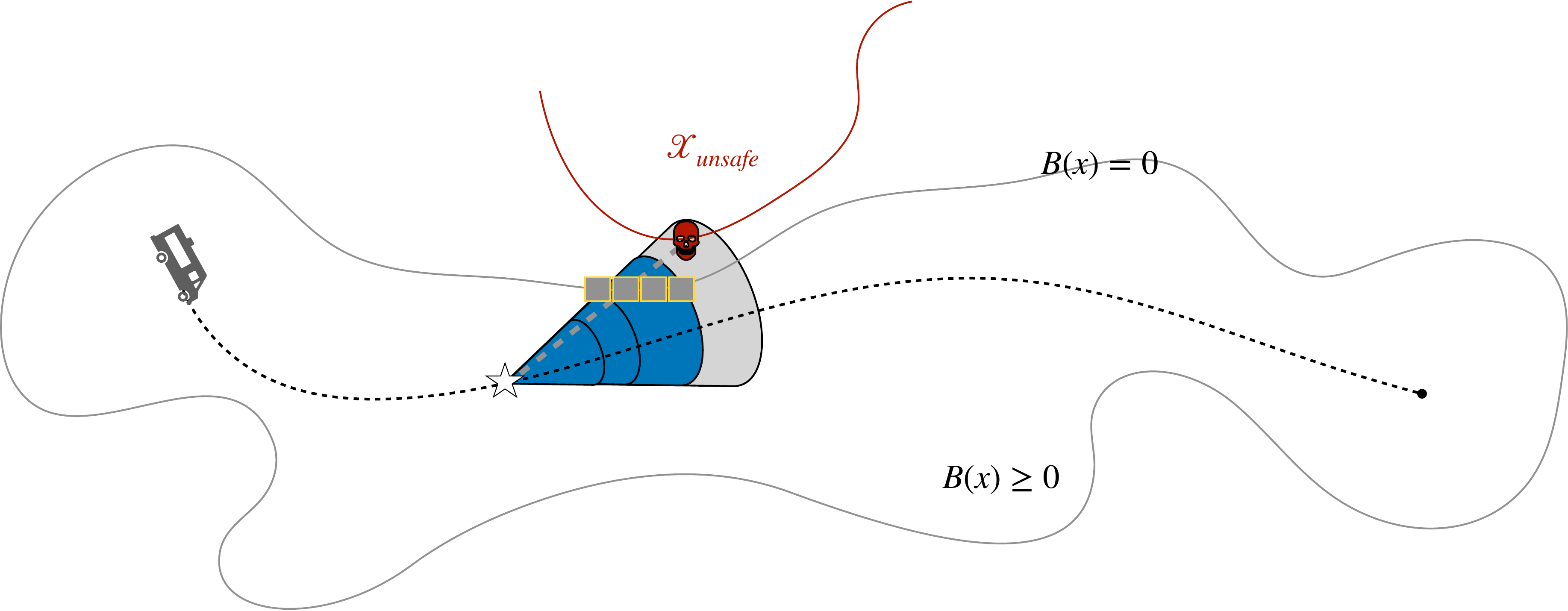}\vspace{-0.5em}
    \caption{Illustration of online verification with lookahead. At each time step, the monitor maintains a cone that over-approximates the reachable states within a fixed horizon. If this cone intersects the unsafe region, the monitor searches for a cube on the barrier boundary where $B(x) = 0$. The cone is then shrunk to contain this boundary, and the identified cube along with its neighbors are verified to assess certificate validity.\label{fig:sketch}}
\end{figure}

As illustrated in Figure~\ref{fig:sketch}, the monitor continuously observes the system's trajectory and incrementally builds a conservative over-approximation of reachable states within a fixed lookahead horizon. If no violations are detected within this cone region, the system is deemed safe for the near future.
If a \emph{safety violation} is detected, e.g., the cone intersects the unsafe region, the monitor performs a refinement step. It searches for a state on the certificate boundary ($B(x)=0$) and identifies the corresponding \emph{activation region} (called a \emph{cube}) in which this state lies. Because ReLU networks induce a partition of the state space into finitely many such cubes, and each cube corresponds to a fixed neuron activation pattern, the CBF condition from Def.~\ref{def:barrier} becomes a set of linear constraints within each cube.
The monitor then verifies these conditions locally by checking whether the barrier certificate holds in the identified cube and its neighbors. If a violation is detected, the system triggers a fail-safe response (e.g., controller override). If all cubes are verified to be safe, the monitor continues observing and updating the cone over time.

Our approach operates over two state spaces: (1) the continuous, real-valued state space defined by the underlying dynamical system, and (2) the discrete space of {activation patterns}, induced by the architecture of the neural CBF. 
In the following, we refer to an element of (1) simply as a \emph{state}, and to the subset of states corresponding to a particular activation pattern as a \emph{cube}. Two states belong to the same cube if they produce identical neuron activation patterns in the CBF. We present the main monitoring loop in Algorithm~\ref{alg:monitor-main}.

\medskip
\noindent\emph{Initialization.} Our approach begins with a pre-trained neural CBF and an associated control policy. The control policy is treated as a black-box, meaning its internal structure is not accessible. However, we assume access to a set of known constraints that bound the range of control inputs generated by the policy, e.g., minimum and maximum allowable acceleration values. In addition, we define a fixed {horizon} \emph{h} parameter, which determines how far into the future the monitor evaluates the system’s evolution to assess safety. 
\begin{algorithm}
\caption{Main Monitor Loop for ReLU-based CBFs}
\label{alg:monitor-main}
\begin{algorithmic}[1]
\State \textbf{Given:}  forward invariant set $\mathcal{C}$, unsafe set $\mathcal{X}_u$, lookahead horizon $h$
\Function{Init}{\,}
\State $v_0\gets 1$; $\mathcal{X}_\text{safe} \gets \aS\setminus \mathcal{X}_u$
\EndFunction
\Function{Next}{$\cS$}
    \State $(cone, x_\text{unsafe}, i) \gets \textsc{ConstructCone}(\mathcal{X}_u, x, h)$
    \textcolor{gray}{// Abstraction function call}
    \If{$x_\text{unsafe}$ is found}  
    \State $v_{\mathrm{buffer}} \gets $ \textsc{VerifyCubesOnBoundary}$(x, x_\text{unsafe}, cone, i, h)$
    \textcolor{gray}{// Verifier call}
    \State $v_k \gets  v_{\mathrm{buffer}}$ \textbf{if} $v_{k-1} = 1$ \textbf{else} $v_k \gets 0$ 
    \State\Return $v_k$
    
    \EndIf
   
\EndFunction
\end{algorithmic}
\end{algorithm}
\begin{algorithm}
\caption{ConstructCone (The Abstraction Function)}
\label{alg:expand}
\begin{algorithmic}[1]
\Function{ConstructCone}{$\mathcal{X}_u$, $x$, $h$}
    \State $cone \gets \{x\}$, $i \gets 0$ \textcolor{gray}{// Initialize cone}
    \While{$i < h$}
        \State $slice \gets \textsc{expand}(cone) \setminus cone$
        \State $cone \gets cone \cup slice$
        \If{$slice \cap \mathcal{X}_u \neq \emptyset$}
            \State \Return $(cone, \textsc{pick}(slice \cap \mathcal{X}_u), i)$
        \EndIf
        \State $i \gets i + 1$
    \EndWhile
    \State \Return $(cone, \bot, h)$
\EndFunction
\end{algorithmic}
\end{algorithm}

\begin{algorithm}
\caption{VerifyCubesOnBoundary (The Verifier)}
\label{alg:verify}
\begin{algorithmic}[1]
\Function{VerifyCubesOnBoundary}{$x$, $x_\text{unsafe}$, $cone$, $i$, $h$}
     \State $cube \gets \textsc{BinarySearch}(x, x_\text{unsafe})$
    \State $boundary \gets \{cube\}$
    \While{$boundary \neq \emptyset$}
        \State $queue \gets \emptyset$
        \ForAll{$c \in boundary$}
            \If{$c \cap \mathcal{X} \cap cone = \emptyset$} \textbf{continue}
            \ElsIf{$\lnot \textsc{VerifyLinear}(c)$} \textcolor{gray}{// Check conditions in Def.~\ref{def:relu}}
                \State \textsc{fail-safe}()
            \Else
                \State $queue.\textsc{push}(c)$
            \EndIf
        \EndFor
        \State $boundary \gets \emptyset$
        \While{$queue \neq \emptyset$}
            \ForAll{$c' \in \textsc{neighborhood}(queue.pop())$}
                \If{$\textsc{Verified}(c')$ or $c' \cap \mathcal{X} = \emptyset$}
                    \textbf{continue}
                \ElsIf{$c' \cap \mathcal{X} \cap cone = \emptyset$}
                    \State $boundary.\textsc{push}(c')$
                \ElsIf{$\lnot \textsc{VerifyLinear}(c')$} 
                    \State \textsc{fail-safe}()
                \Else
                    \State $queue.\textsc{push}(c')$
                \EndIf
            \EndFor
        \EndWhile
        \State $cone \gets cone \cup \textsc{expand}(cone)$
        \State $i \gets i + 1$
        \If{$i > h$} \textbf{break} \EndIf
    \EndWhile
\EndFunction
\end{algorithmic}
\end{algorithm}

During monitoring, the cone is constructed using an abstraction function. To detect potential safety violations, the monitor expands the cone using the \textsc{ConstructCone} procedure (Algorithm~\ref{alg:expand}). If no intersection with the unsafe set $\mathcal{X}_u$ is found within the horizon, the system continues unimpeded. However, if an unsafe state is detected, the monitor performs a binary search between the current state and the unsafe state to locate a cube on the barrier boundary (line 2, Algorithm~\ref{alg:verify}). 
The online verification procedure begins by adding the identified cube to a queue, referred to as the \emph{boundary}. 
Each cube in this queue is then iteratively verified. Cubes that do not intersect with the current lookahead cone are discarded, as they cannot influence the set of reachable states. If verification fails for any cube, this indicates a violation of the certificate conditions, prompting the system to enter a fail-safe mode or switch to a backup controller. Successfully verified cubes are marked accordingly and added to a queue for neighborhood expansion.

With this initial queue, we start a breadth-first search by examining the 1-bit Hamming neighbors of each verified cube, i.e., those whose ReLU activation patterns differ by exactly one neuron. A neighboring cube is discarded if it has already been verified, or does not intersect any state in $\mathcal{X}$. If it lies outside the current cone, we may verify it after expanding the cone at the next iteration.
All intersection checks as well as \textsc{VerifyLinear} in Algorithm~\ref{alg:verify} operate on a single cube, thus they can be solved efficiently as a simple set of linear inequalities.

\section{Experiments}
In this section, we evaluate the effectiveness of our monitoring method via a case study and compare it with static verification. 

The benchmark we used is a satellite rendezvous example, adapted from~\cite{dawson2022safe,jewison2016spacecraft}. In this scenario, a chaser satellite attempts to approach a target satellite while remaining within a designated safe region, corresponding to its line of sight. Both satellites are assumed to be in orbit around the Earth.
The system state is represented by the 6-dimensional vector $[x, y, z, v_x, v_y, v_z]$, where $[x, y, z]$ denotes the relative position of the chaser with respect to the target, and $[v_x, v_y, v_z]$ its relative velocity. The control input is given by $u = [u_x, u_y, u_z]$, representing thrust applied along each axis. The control interval is fixed to be $0.1$ second.

We evaluate both static formal verification and online monitoring across different neural CBF architectures. Static verification is performed using SEEV~\cite{DBLP:conf/nips/ZhangQG024}, a recent tool tailored for verifying ReLU-based certificates. For online monitoring, we employ a local abstraction function that is fixed-step unrolling with known input bounds and a verifier operating on the ReLU activation patterns as described in Section~5.
\begin{table}[h]
\centering
\caption{Verification results with corresponding full verification times using SEEV~\cite{DBLP:conf/nips/ZhangQG024}.}
\begin{tabular}{c|c|c|r}
\toprule
No. Hidden Layers & No. Neurons per Layer & Verification Result & Full Time (s) \\
\midrule

2  & 8  & safe & 50.21 \\
4  & 8  & safe & 253.20 \\
8  & 16 & -- &  \textcolor{darkred}{>2 hours}\\
16 & 16 & unsafe & 6.03 \\
\bottomrule
\end{tabular}

\label{tab:verification}
\end{table}

\begin{figure}
    \centering
    \includegraphics[width=\linewidth]{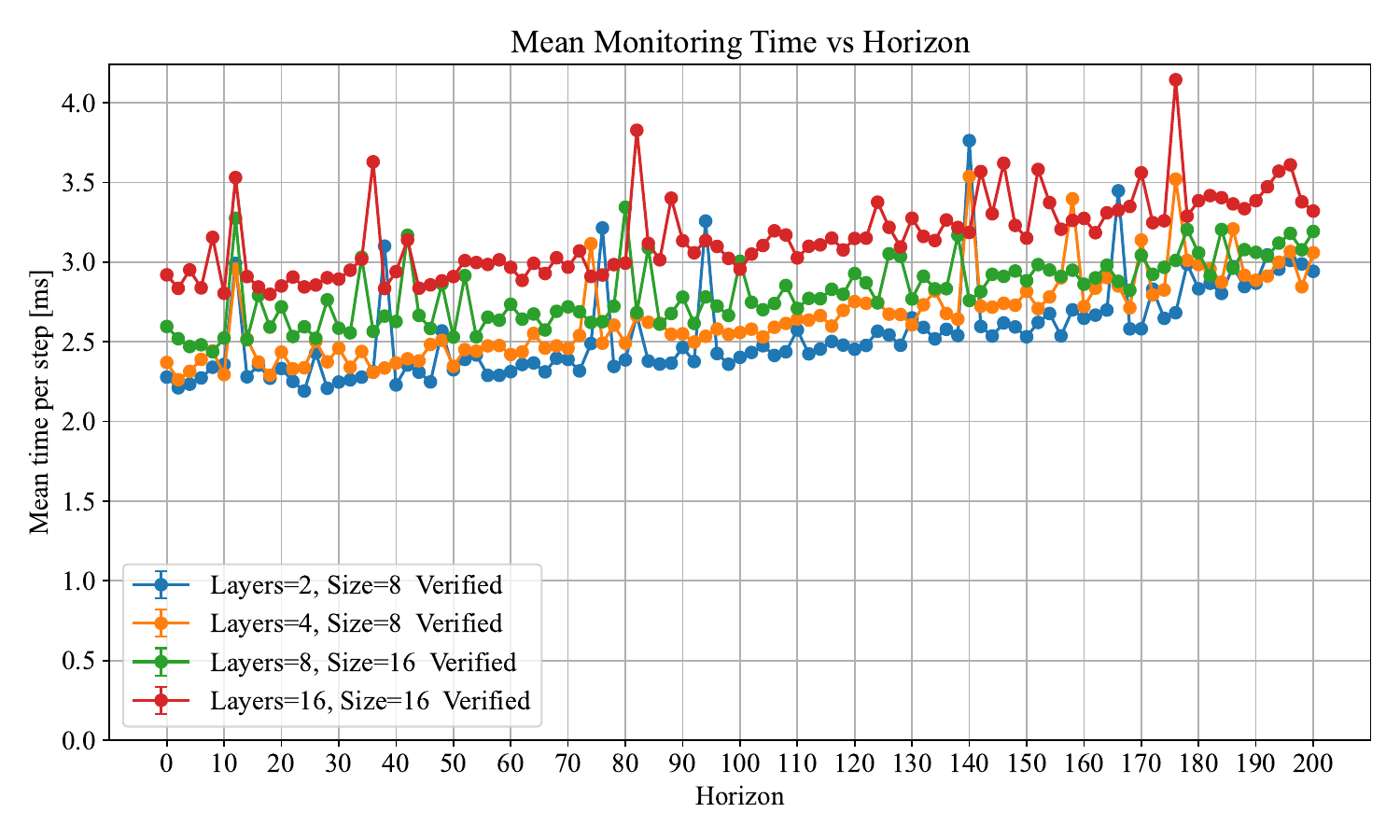}
    \includegraphics[width=\linewidth]{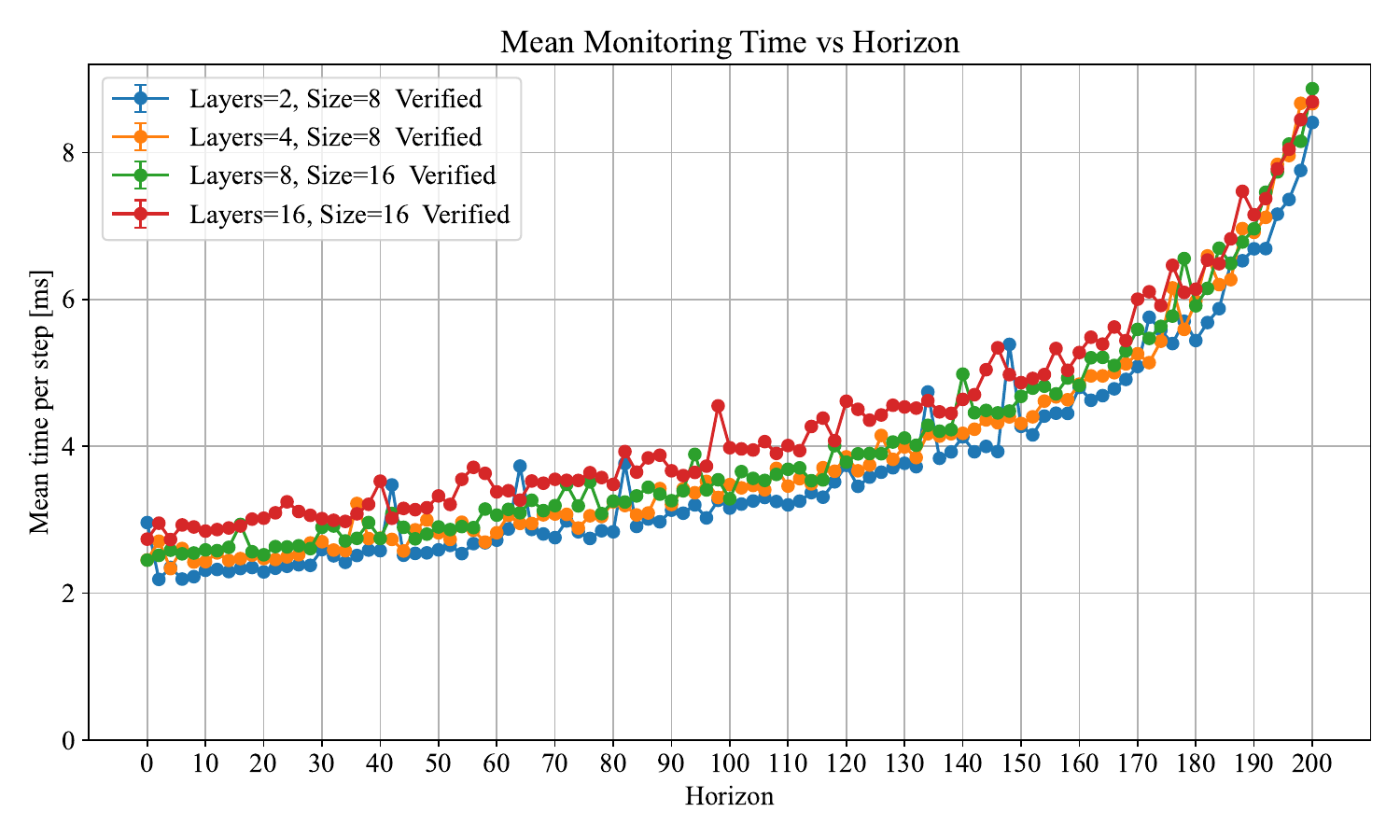}
    \caption{Monitoring overhead measured for lookahead horizons up to 200 steps, across four network configurations with varying depth and width. All configurations were successfully verified through static analysis (see Table~\ref{tab:verification}). The two plots correspond to different initial states. The monitor did not raise warnings during these runs.}
    \label{fig1}
\end{figure}

Table~\ref{tab:verification} presents the results of static formal verification for neural CBFs of varying sizes. Notably, the verification of CBF with 8 layers and 16 neurons per layer timed out after two hours using the SEEV method. In contrast, our online monitoring approach was able to detect violations much faster (see Figure~\ref{fig2}).

We now evaluate how our monitoring framework performs. As shown in Figure~\ref{fig2}, the monitor was able to detect violations with minimal overhead with a longer lookahead horizon (> 70 steps). With shorter horizons, no violations were detected, as the future cones remained relatively small. This also give an insight: even incorrect certificates may suffice to ensure safety over limited parts of the system's state space. As the lookahead horizon increases, the monitor explores a broader region and is more likely to encounter violations where the certificate conditions fail. We also observe an increase in monitoring overhead around 70 steps. This is because, as the horizon increases, the over-approximation increases in size, thus intersects more barrier cubes which need to be verified. 
\begin{figure}
    \centering
    \includegraphics[width=\linewidth]{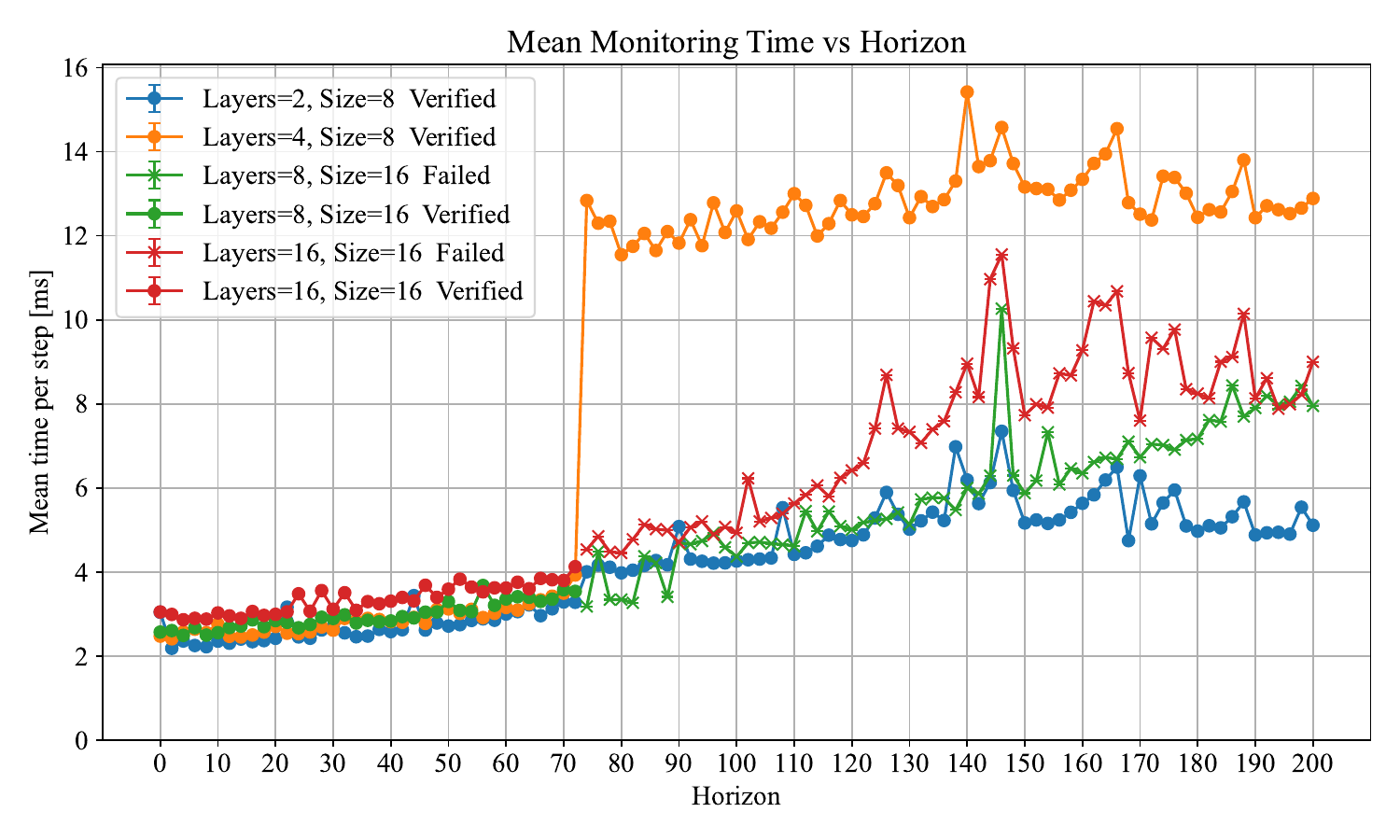}\vspace{-1em}
    \caption{Monitoring overhead measured for lookahead horizons up to 200 steps, across four network configurations with varying depth and width. Two CBFs failed static verification (one of them timed out), and our monitor successfully detected violations online in both cases, with a higher horizon. }
    \label{fig2}
\end{figure}
\begin{table}[ht]
\centering
\caption{Online verification results for various lookahead horizons. The outcome column indicates whether verification failed or succeeded. We also display the number of traces evaluated (\#traces), with a total number of 323 traces. }
\label{tab:monitoring}
\renewcommand{\arraystretch}{1.2}
\setlength{\tabcolsep}{6pt}

\begin{tabular}{ccccccccc}
\toprule
\multicolumn{4}{c}{} & \multicolumn{5}{c}{Lookahead Horizon [steps]} \\
\cmidrule(lr){5-9}
Layers & Size & Outcome & Metric & 40 & 80 & 120 & 160 & 200 \\
\midrule
2  & 8  & \checkmark & \#traces & 323 & 323 & 323 & 323 & 323 \\
   &    &            & time [ms]     & 2.52 & 2.85 & 3.27 & 3.82 & 4.06 \\
\midrule
4  & 8  & \checkmark & \#traces & 323 & 323 & 323 & 323 & 323 \\
   &    &            & time [ms]     & 2.72 & 3.64 & 4.90 & 6.69 & 8.21 \\
\midrule
8  & 16 & \checkmark & \#traces & 318 & 292 & 255 & 199 & 142 \\
   &    &            & time [ms]     & 2.76 & 3.01 & 3.35 & 3.75 & 4.01 \\
   \midrule
8  & 16 & \ding{55}  & \#traces & 5 & 26 & 63 & 113 & 164 \\
   &     &           & time [ms]     & 6.18 & 5.09 & 5.18 & 8.64 & 7.84 \\
\midrule
16 & 16 & \checkmark & \#traces & 318 & 292 & 255 & 199 & 142 \\
   &     &           & time [ms]     & 3.09 & 3.33 & 3.70 & 4.07 & 4.39 \\

\midrule
16 & 16 & \ding{55}  & \#traces & 5 & 31 & 68 & 124 & 181 \\
   &     &           & time [ms]     & 6.49 & 5.70 & 6.47 & 10.37 & 8.93 \\
\bottomrule
\end{tabular}
\end{table}

There is a trade-off between monitoring overhead and lookahead horizon, which is as expected. However, we observe that the overhead remains well within practical limits. Figures~\ref{fig1} and~\ref{fig2} report average monitoring overheads for lookahead horizons up to 200 steps (20 seconds). Across all six network configurations, the average overhead remains below 16ms per step, which is substantially lower than the 0.1s control interval. This demonstrates the practicality of our approach for real-time safety monitoring.

In Table~\ref{tab:monitoring}, we report the results of online monitoring across different network architectures with lookahead horizons. For each configuration, we summarize the number of traces evaluated and the average runtime per monitoring step. Each configuration was tested on a total of 323 traces, generated by random initial states. 
As the lookahead horizon increases, our monitor successfully detects more violations in networks that failed verification, confirming the effectiveness of lookahead in revealing unsafe behaviors. The per-step verification overhead remains low across all configurations, even for larger networks.

In addition to deployment-time use, our monitor is also well-suited for identifying counterexamples \emph{during testing}. This makes it a valuable tool for diagnosing safety and certificate violations early in the development cycle. Moreover, it can be integrated into a learner–monitor loop for iteratively repairing certificates, as proposed in~\cite{yu2025neural}. This can be achieved by detecting violations online and adding counterexamples to the training data.

\medskip
\noindent\emph{Discussion and limitations.} We consider our online monitoring framework to be complementary to static verification for both testing and deployment settings. This can also be particularly valuable in the presence of environmental uncertainties or model inaccuracies, where exhaustive offline verification may not capture all relevant behaviors. While in our case, the monitor still uses the system model for prediction, it performs verification adaptively along the concrete execution path, avoiding exhaustive offline analysis over the entire state space. A potential limitation arises when the system operates near the boundary of the certificate region. In such cases, verifying large numbers of cubes may introduce additional computational overhead, which can affect real-time performance. Moreover, if the system is already close to exiting the forward-invariant set, a fail-safe mechanism may not always react in time to prevent a safety violation. This is a fundamental challenge of runtime monitoring in general, where monitors are expected to signal warnings timely while minimizing false alarms. Finally, while our case study focuses on a linearized orbital rendezvous model and a corresponding verification procedure, the proposed framework is general. It applies to nonlinear control-affine systems, provided that suitable local abstractions and verifiers are available. The frameworks capabilities scale with the performance of either component.

\section{Related Work}
\emph{Black-box simplex architecture.} The Simplex architecture is a runtime assurance method that has been widely adopted in control applications~\cite{desai2019soter,phan2017collision,schierman2015runtime}. It consists of an active controller responsible for primary decision-making and a backup controller that can override the active controller when necessary to ensure runtime safety. Recent research has extended this framework to black-box settings, where the internal architecture of the controllers is assumed to be unknown~\cite{maderbacher2023provable,mehmood2022black}. This is a similar setting to our work where we perform on-the-fly verification, however, they do not consider certificates.

\medskip
\noindent
\emph{Runtime enforcement using certificates.} Shielding is a widely adopted runtime enforcement technique~\cite{bloem2015shield} where the shield overrides control inputs to ensure safety. 
Certificate functions, such as barrier functions, have been integrated with shielding, as demonstrated in~\cite{yin2023shield}, where the barrier functions are given in advance. However, in that setting the certificates are not formally verified thus there is no correctness guarantee. More recently, Corsi et al.~\cite{mandal2024formally} propose verification-guided shielding, which uses pre-computed verification to identify specific regions for shielding. Unlike these approaches, our method performs real-time verification over a lookahead horizon without requiring preprocessing.

\medskip
\noindent
\emph{Concolic Testing.} Related to our work is also concolic testing~\cite{sen2007concolic,DBLP:conf/hvc/Sen09,DBLP:conf/sigsoft/SenMA05}, which combines concrete and symbolic execution to explore the local state space of software programs more effectively. Our work draws a parallel to this idea by applying a similar concolic approach in the context of certificate-based control.

\medskip
\noindent
\emph{Formal verification of certificates.} 
To verify a neural certificate, we need to check whether the network satisfies certificate conditions w.r.t.\ the controlled system. Broadly speaking there are two approaches, a symbolic approach and a search approach.
In the symbolic approach the system, the network, and the certificate condition are encoded as a logical formula and solved using a SMT solver~\cite{zhao2020synthesizing,giacobbeneural,abate2021fossil}, or they are encoded as a mixed integer linear program and solved with a solver such as Gurobi~\cite{zhao2022verifying,gurobi}.
In the search approach the Lipschitz constant of a neural network is exploited to systematically search the state space, e.g., grid search, to detect potential certificate violations~\cite{DBLP:conf/tacas/ChatterjeeHLZ23,tayal2024learning}.  

\section{Conclusion}
In this work, we proposed a general framework for online verification of neural certificates that integrates partial static verification over a finite lookahead horizon. Our method is lightweight, does not require access to future control inputs, and supports modular integration of local abstraction function and certificate verifiers. This enables runtime safety assurance even in settings where static verification is infeasible.
We instantiated our framework with a verifier tailored to ReLU-based control barrier functions and demonstrated its practicality on a satellite rendezvous example. The monitor was able to effectively detect certificate violations as well as ensuring runtime safety with minimal overhead.

Our framework is general and readily applies to other types of certificate functions, such as Lyapunov functions and contraction metrics~\cite{sun2021learning}, and opens the door to further integration with runtime enforcement mechanisms such as shielding. As future work, we plan to investigate how this monitoring framework can be embedded into training loops to iteratively repair unsafe certificates. We also plan to extend our approach to probabilistic settings to account for environmental uncertainties.

\subsection*{Acknowledgement} 
This
work is supported by the European Research Council under Grant No.: ERC-2020-AdG 101020093.
\bibliographystyle{splncs04}
\bibliography{refs}

\begin{thebibliography}{10}
\providecommand{\url}[1]{\texttt{#1}}
\providecommand{\urlprefix}{URL }
\providecommand{\doi}[1]{https://doi.org/#1}

\bibitem{abate2021fossil}
Abate, A., Ahmed, D., Edwards, A., Giacobbe, M., Peruffo, A.: Fossil: a software tool for the formal synthesis of lyapunov functions and barrier certificates using neural networks. In: Proceedings of the 24th international conference on hybrid systems: computation and control. pp. 1--11 (2021)

\bibitem{abiodun2018state}
Abiodun, O.I., Jantan, A., Omolara, A.E., Dada, K.V., Mohamed, N.A., Arshad, H.: State-of-the-art in artificial neural network applications: A survey. Heliyon  \textbf{4}(11) (2018)

\bibitem{ames2014control}
Ames, A.D., Grizzle, J.W., Tabuada, P.: Control barrier function based quadratic programs with application to adaptive cruise control. In: 53rd IEEE conference on decision and control. pp. 6271--6278. IEEE (2014)

\bibitem{anwar2018medical}
Anwar, S.M., Majid, M., Qayyum, A., Awais, M., Alnowami, M., Khan, M.K.: Medical image analysis using convolutional neural networks: a review. Journal of medical systems  \textbf{42},  1--13 (2018)

\bibitem{bloem2015shield}
Bloem, R., K{\"o}nighofer, B., K{\"o}nighofer, R., Wang, C.: Shield synthesis: Runtime enforcement for reactive systems. In: International conference on tools and algorithms for the construction and analysis of systems. pp. 533--548. Springer (2015)

\bibitem{chang2019neural}
Chang, Y.C., Roohi, N., Gao, S.: Neural lyapunov control. Advances in neural information processing systems  \textbf{32} (2019)

\bibitem{DBLP:conf/tacas/ChatterjeeHLZ23}
Chatterjee, K., Henzinger, T.A., Lechner, M., Zikelic, D.: A learner-verifier framework for neural network controllers and certificates of stochastic systems. In: {TACAS} {(1)}. Lecture Notes in Computer Science, vol. 13993, pp. 3--25. Springer (2023)

\bibitem{crenshaw2007simplex}
Crenshaw, T.L., Gunter, E., Robinson, C.L., Sha, L., Kumar, P.: The simplex reference model: Limiting fault-propagation due to unreliable components in cyber-physical system architectures. In: 28th IEEE International Real-Time Systems Symposium (RTSS 2007). pp. 400--412. IEEE (2007)

\bibitem{DBLP:journals/trob/DawsonGF23}
Dawson, C., Gao, S., Fan, C.: Safe control with learned certificates: {A} survey of neural lyapunov, barrier, and contraction methods for robotics and control. {IEEE} Trans. Robotics  \textbf{39}(3),  1749--1767 (2023)

\bibitem{dawson2022safe}
Dawson, C., Qin, Z., Gao, S., Fan, C.: Safe nonlinear control using robust neural lyapunov-barrier functions. In: Conference on Robot Learning. pp. 1724--1735. PMLR (2022)

\bibitem{desai2019soter}
Desai, A., Ghosh, S., Seshia, S.A., Shankar, N., Tiwari, A.: Soter: a runtime assurance framework for programming safe robotics systems. In: 2019 49th Annual IEEE/IFIP International Conference on Dependable Systems and Networks (DSN). pp. 138--150. IEEE (2019)

\bibitem{giacobbeneural}
Giacobbe, M., Kroening, D., Pal, A., Tautschnig, M.: Neural model checking. In: The Thirty-eighth Annual Conference on Neural Information Processing Systems (2024)

\bibitem{gurobi}
{Gurobi Optimization, LLC}: {Gurobi Optimizer Reference Manual} (2024), \url{https://www.gurobi.com}

\bibitem{DBLP:journals/corr/abs-2412-16564}
Henzinger, T.A., Kresse, F., Mallik, K., Yu, E., Zikelic, D.: Predictive monitoring of black-box dynamical systems. In: {L4DC}. Proceedings of Machine Learning Research, {PMLR} (2025)

\bibitem{jewison2016spacecraft}
Jewison, C., Erwin, R.S.: A spacecraft benchmark problem for hybrid control and estimation. In: 2016 IEEE 55th Conference on Decision and Control (CDC). pp. 3300--3305. Ieee (2016)

\bibitem{kiumarsi2017optimal}
Kiumarsi, B., Vamvoudakis, K.G., Modares, H., Lewis, F.L.: Optimal and autonomous control using reinforcement learning: A survey. IEEE transactions on neural networks and learning systems  \textbf{29}(6),  2042--2062 (2017)

\bibitem{maderbacher2023provable}
Maderbacher, B., Schupp, S., Bartocci, E., Bloem, R., Ni{\v{c}}kovi{\'c}, D., K{\"o}nighofer, B.: Provable correct and adaptive simplex architecture for bounded-liveness properties. In: International Symposium on Model Checking Software. pp. 141--160. Springer (2023)

\bibitem{mandal2024formally}
Mandal, U., Amir, G., Wu, H., Daukantas, I., Newell, F.L., Ravaioli, U.J., Meng, B., Durling, M., Ganai, M., Shim, T., et~al.: Formally verifying deep reinforcement learning controllers with lyapunov barrier certificates. In: \# PLACEHOLDER\_PARENT\_METADATA\_VALUE\#. pp. 95--106. TU Wien Academic Press (2024)

\bibitem{mehmood2022black}
Mehmood, U., Sheikhi, S., Bak, S., Smolka, S.A., Stoller, S.D.: The black-box simplex architecture for runtime assurance of autonomous cps. In: NASA formal methods symposium. pp. 231--250. Springer (2022)

\bibitem{peruffo2021automated}
Peruffo, A., Ahmed, D., Abate, A.: Automated and formal synthesis of neural barrier certificates for dynamical models. In: International conference on tools and algorithms for the construction and analysis of systems. pp. 370--388. Springer (2021)

\bibitem{phan2017collision}
Phan, D., Yang, J., Grosu, R., Smolka, S.A., Stoller, S.D.: Collision avoidance for mobile robots with limited sensing and limited information about moving obstacles. Formal Methods in System Design  \textbf{51},  62--86 (2017)

\bibitem{DBLP:journals/tac/PrajnaJP07}
Prajna, S., Jadbabaie, A., Pappas, G.J.: A framework for worst-case and stochastic safety verification using barrier certificates. {IEEE} Trans. Autom. Control.  \textbf{52}(8),  1415--1428 (2007)

\bibitem{schierman2015runtime}
Schierman, J.D., DeVore, M.D., Richards, N.D., Gandhi, N., Cooper, J.K., Horneman, K.R., Stoller, S., Smolka, S.: Runtime assurance framework development for highly adaptive flight control systems. Barron Associates, Inc. Charlottesville, Tech. Rep  (2015)

\bibitem{sen2007concolic}
Sen, K.: Concolic testing. In: Proceedings of the 22nd IEEE/ACM international conference on Automated software engineering. pp. 571--572 (2007)

\bibitem{DBLP:conf/hvc/Sen09}
Sen, K.: {DART:} directed automated random testing. In: Haifa Verification Conference. Lecture Notes in Computer Science, vol.~6405, p.~4. Springer (2009)

\bibitem{DBLP:conf/sigsoft/SenMA05}
Sen, K., Marinov, D., Agha, G.: {CUTE:} a concolic unit testing engine for {C}. In: {ESEC/SIGSOFT} {FSE}. pp. 263--272. {ACM} (2005)

\bibitem{seto1998simplex}
Seto, D., Krogh, B., Sha, L., Chutinan, A.: The simplex architecture for safe online control system upgrades. In: Proceedings of the 1998 American Control Conference. ACC (IEEE Cat. No. 98CH36207). vol.~6, pp. 3504--3508. IEEE (1998)

\bibitem{DBLP:journals/automatica/Smith95}
Smith, M.C.: The general problem of the stability of motion : Translated and edited by a. t. fuller. taylor and francis, 1992. Autom.  \textbf{31}(2),  353--354 (1995)

\bibitem{sun2021learning}
Sun, D., Jha, S., Fan, C.: Learning certified control using contraction metric. In: conference on Robot Learning. pp. 1519--1539. PMLR (2021)

\bibitem{tayal2024learning}
Tayal, M., Zhang, H., Jagtap, P., Clark, A., Kolathaya, S.: Learning a formally verified control barrier function in stochastic environment. arXiv preprint arXiv:2403.19332  (2024)

\bibitem{yin2023shield}
Yin, J., Dawson, C., Fan, C., Tsiotras, P.: Shield model predictive path integral: A computationally efficient robust mpc method using control barrier functions. IEEE Robotics and Automation Letters  \textbf{8}(11),  7106--7113 (2023)

\bibitem{yu2025neural}
Yu, E., {\v{Z}}ikeli{\'c}, {\DJ}., Henzinger, T.A.: Neural control and certificate repair via runtime monitoring. In: Proceedings of the AAAI Conference on Artificial Intelligence. vol.~39, pp. 26409--26417 (2025)

\bibitem{DBLP:conf/nips/ZhangQG024}
Zhang, H., Qin, Z., Gao, S., Clark, A.: {SEEV:} synthesis with efficient exact verification for relu neural barrier functions. In: NeurIPS (2024)

\bibitem{zhang2023exact}
Zhang, H., Wu, J., Vorobeychik, Y., Clark, A.: Exact verification of relu neural control barrier functions. Advances in neural information processing systems  \textbf{36},  5685--5705 (2023)

\bibitem{zhao2020synthesizing}
Zhao, H., Zeng, X., Chen, T., Liu, Z.: Synthesizing barrier certificates using neural networks. In: Proceedings of the 23rd international conference on hybrid systems: Computation and control. pp. 1--11 (2020)

\bibitem{zhao2022verifying}
Zhao, Q., Chen, X., Zhao, Z., Zhang, Y., Tang, E., Li, X.: Verifying neural network controlled systems using neural networks. In: Proceedings of the 25th ACM International Conference on Hybrid Systems: Computation and Control. pp. 1--11 (2022)

\end{thebibliography}
\end{document}